\documentclass[a4paper,12pt]{article}
\usepackage{graphicx}
\usepackage{epsfig}
\usepackage{bbding}
\usepackage[T1]{fontenc}
\usepackage{bbold}

\usepackage{amsmath,amssymb,amsthm}

\usepackage{float}
\usepackage{wrapfig}
\usepackage{subcaption} % side by side figures
\usepackage[export]{adjustbox}  % right/left figures

\usepackage{algorithmicx}
\usepackage{algpseudocode}
\usepackage{algorithm}

\newcommand{\abs}[1]{\left|{#1}\right|}
\newcommand{\ceil}[1]{\left\lceil{#1}\right\rceil}
\newcommand{\floor}[1]{\left\floor{#1}\right\floor}
\newcommand{\DS}{{\mathop {\rm DS}}}
\newcommand{\BB}{{\mathcal{B}}}
\newcommand{\CC}{{\mathcal{C}}}
\newcommand{\PP}{{\mathcal{P}}}
\newcommand{\QQ}{{\mathcal{Q}}}

\newtheorem{lemma}{Lemma}%[section]
\newtheorem{theorem}{Theorem}%[section]
\newtheorem*{theorem1}{Theorem 1}%[section]
\newtheorem{property}{Property}%[section]
\newtheorem{definition}{Definition}%[section]

\begin{document}

\begin{center}
{\Large Constant Query Time $(1+\epsilon)$-Approximate Distance Oracle for 
Planar Graphs\footnote{A preliminary version of the paper appeared in the
Proceedings of the 26th International Symposium on Algorithms and Computation
(ISAAC 2015) \cite{GX15}}}.
\vskip 0.2in

Qian-Ping Gu and Gengchun Xu

School of Computing Science, Simon Fraser University\\
Burnaby BC Canada V5A1S6\\
qgu@cs.sfu.ca,gxa2@sfu.ca
\end{center}

\noindent
{\bf Abstract:}
We give a $(1+\epsilon)$-approximate distance oracle with $O(1)$ query time for 
an undirected planar graph $G$ with $n$ vertices and non-negative edge lengths. 
For $\epsilon>0$ and any two vertices $u$ and $v$ in $G$, our oracle gives a 
distance $\tilde{d}(u,v)$ with stretch $(1+\epsilon)$ in $O(1)$ time. The oracle 
has size $O(n\log n ((\log n)/\epsilon+f(\epsilon)))$ and pre-processing time 
$O(n\log n((\log^3 n)/\epsilon^2+f(\epsilon)))$, where 
$f(\epsilon)=2^{O(1/\epsilon)}$. This is the first $(1+\epsilon)$-approximate 
distance oracle with $O(1)$ query time independent of $\epsilon$ and the size and 
pre-processing time nearly linear in $n$, and improves the query time 
$O(1/\epsilon)$ of previous $(1+\epsilon)$-approximate distance oracle with size 
nearly linear in $n$.

\noindent
{\bf key words:} Distance oracle, planar graphs, approximate algorithms, graph
decomposition

\section{Introduction}
\label{intro}

Finding a distance between two vertices in a graph is a fundamental computational 
problem and has a wide range of applications. For this problem, there is a rich 
literature of algorithms. This problem can be solved by a single source shortest 
path algorithm such as the Dijkstra and Bellman-Ford algorithms. In many 
applications, it is required to compute the shortest path distance in an extreme 
short time. One approach to meet such a requirement is to use distance oracles.

A distance oracle is a data structure which keeps the pre-computed distance
information and provides a distance between any given pair of vertices very 
efficiently. There are two phases in the distance oracle approach. The first 
phase is to compute the data structure for a given graph $G$ and the second 
is to provide an answer for a query on the distance between a pair of vertices 
in $G$. The efficiency of distance oracles is mainly measured by the time to 
answer a query ({\em query time}), the memory space required for the data 
structure ({\em oracle size}) and the time to create the data structure 
({\em pre-processing time}). Typically, there is a trade-off between the query 
time and the oracle size. A simple approach to compute a distance oracle for 
graph $G$ of $n$ vertices is to solve the all pairs shortest paths problem in $G$ 
and keep the shortest distances in an $n\times n$ distance array. This gives an 
oracle with $O(1)$ query time and $O(n^2)$ size. A large number of papers have 
been published for distance oracles with better measures on the product of query 
time and oracle size, see Sommer's paper for a survey \cite{Som12}.

Planar graphs are an important model for many networks such as the road networks.
Distance oracles for planar graphs have been extensively studied. Djidjev proves 
that for any oracle size $S\in [n,n^2]$, there is an exact distance oracle with 
query time $O(n^2/S)$ for weighted planar graphs \cite{Dji96}. There are several 
exact distance oracles with size $O(S)$ and more efficient query time for different
ranges of $S$, for example, an oracle by Wulff-Nilsen \cite{WN10a} with $O(1)$ 
query time and $O(n^2(\log\log n)^4/\log n)$ size for weighted directed planar 
graphs and an oracle by Mozes and Sommer \cite{MS12} with query time 
$O((n/\sqrt{S})\log^{2.5}n)$ and size $S\in [n\log\log n,n^2]$ for weighted
directed planar graph. Readers may refer to Sommer's survey paper \cite{Som12} 
for more details.

Approximate distance oracles have been developed to achieve very fast query time
and near linear size for planar graphs. For vertices $u$ and $v$ in graph $G$, 
let $d_G(u,v)$ denote the distance between $u$ and $v$. An oracle is called an 
$\alpha$-approximate oracle or with {\em stretch $\alpha$} for $\alpha\geq 1$ if 
it provides a distance $\tilde{d}(u,v)$ with
$d_G(u,v)\leq \tilde{d}(u,v)\leq \alpha d_G(u,v)$ for $u$ and $v$ in $G$. An 
oracle is said to have an {\em additive stretch $\beta\ge 0$} if it provides a 
distance $\tilde{d}(u,v)$ with $d_G(u,v)\leq \tilde{d}(u,v)\leq d_G(u,v)+\beta$. 
For $\epsilon>0$, Thorup gives a $(1+\epsilon)$-approximate distance oracle with 
$O(1/\epsilon)$ (resp. $O(1/\epsilon +\log\log\Delta)$, where $\Delta$ is the 
longest finite distance between any pair of vertices in $G$) query time and 
$O(n\log n/\epsilon)$ (resp. $O(n(\log \Delta)\log n/\epsilon)$) size for an
undirected (resp. directed) planar $G$ with non-negative edge lengths \cite{Tho04}.
A similar result for undirected planar graphs is found independently by Klein
\cite{Klein02}. Kawarabayashi et al. give a $(1+\epsilon)$-approximate distance 
oracle with $O((1/\epsilon)\log^2(1/\epsilon)\log\log (1/\epsilon)\log^* n)$ 
query time and $O(n\log n \log\log(1/\epsilon)\log^* n)$ size for undirected 
planar graphs with non-negative edge lengths \cite{KST13}. The query times of the 
oracles above are fast but still at least $O(1/\epsilon)$. Recently, Wulff-Nilsen 
gives a $(1+\epsilon)$-approximate distance oracle with 
$O(n(\log\log n)^2/\epsilon+\log\log n/\epsilon^2))$ size and 
$O((\log\log n)^3/\epsilon^2+\log\log n
\sqrt{\log\log((\log\log n)/\epsilon^2)}/\epsilon^2)$ query time for undirected 
planar graph with non-negative edge lengths \cite{WN16}. This  result has a better 
trade-off between the query time and the oracle size in the size of graph than 
those in \cite{KST13,Klein02,Tho04}.

Distance oracles with constant query time are of both theoretical and practical
importance \cite{Che14,Det14}. Our main result is an $O(1)$ query time 
$(1+\epsilon)$-approximate distance oracle for undirected planar graphs with 
non-negative edge lengths.
\begin{theorem}
Let $G$ be an undirected planar graph with $n$ vertices and non-negative edge 
lengths and let $\epsilon >0$. There is a $(1+\epsilon)$-approximate distance 
oracle for $G$ with $O(1)$ query time, $O(n\log n(\log n/\epsilon+f(\epsilon)))$ 
size and $O(n\log n(\log^3 n/\epsilon^2+f(\epsilon)))$ pre-processing time, 
where $f(\epsilon)=2^{O(1/\epsilon)}$.
\label{theo-1}
\end{theorem}
The oracle in Theorem~\ref{theo-1} has a constant query time independent of 
$\epsilon$ and size nearly linear in the graph size. This improves the query 
time of the previous works \cite{KST13,Tho04} that are (nearly) linear in 
$1/\epsilon$ for non-constant $\epsilon$. Wulff-Nilsen gives an $O(1)$ time exact 
distance oracle for $G$ with size $O(n^2(\log\log n)^4/\log n)$ \cite{WN10a}. 
There exists some constant $c_0>0$ such that for $\frac{1}{\epsilon}<c_0 \log n$,
%+4\log\log\log n-2\log\log n)$, 
our oracle has a smaller size.

The result in Theorem~\ref{theo-1} can be generalized to an oracle described in
the next theorem.
\begin{theorem}
Let $G$ be an undirected planar graph with $n$ vertices and non-negative edge
lengths, $\epsilon >0$ and $1\leq \eta \leq 1/\epsilon$. There is a 
$(1+\epsilon)$-approximate distance oracle for $G$ with $O(\eta)$ query time, 
$O(n\log n(\log n/\epsilon+f(\eta\epsilon)))$ size and 
$O(n\log n(\log^3 n/\epsilon^2+f(\eta\epsilon)))$ pre-processing time, where
$f(\eta\epsilon)=2^{O(1/(\eta\epsilon))}$.
\label{theo-2}
\end{theorem}

Our results build on some techniques used in the previous approximate distance 
oracles for planar graphs. Thorup \cite{Tho04} gives a $(1+\epsilon)$-approximate 
distance oracle for planar graph $G$ with $O(1/\epsilon)$ query time. Informally,
some techniques used in the oracle are as follows: Decompose $G$ into a balanced 
recursive subdivision; $G$ is decomposed into subgraphs of balanced sizes by 
shortest paths and each subgraph is decomposed recursively until every subgraph 
is reduced to a pre-defined size. A path $Q$ {\it intersects} a path $Q^\prime$ 
if $V(Q)\cap V(Q^\prime)\neq\emptyset$. A set $\QQ$ of paths is a 
path-separator for vertices $u$ and $v$ if every path between $u$ and $v$ 
intersects a path $Q\in \QQ$. Vertices $u$ and $v$ are 
{\it shortest-separated} by a path $Q$ if there exist a shortest path between 
$u$ and $v$ that intersects $Q$. If vertices $u$ and $v$ have a path-separator 
$\QQ$ then $u$ and $v$ is shortest-separated by some path $Q\in\QQ$.
For each subgraph $X$ of $G$, let $\PP(X)$ be the set of shortest paths 
used to decompose $X$. For each path $Q\in \PP(X)$ and each vertex $u$ in 
$X$, a set $P_Q(u)$ of $O(1/\epsilon)$ vertices called {\em portals} on $Q$ is 
selected. For vertices $u$ and $v$ shortest-separated by some path $Q$ in 
$\PP(X)$, 
$\min_{p\in P_Q(u),q\in P_Q(v),Q\in \PP(X)} d_G(u,p)+d_G(p,q)+d_G(q,v)$
is used to approximate $d_G(u,v)$. The oracle keeps the distances $d_G(u,p)$
and $d_G(p,v)$. 

The portal set $P_Q(u)$ above is vertex dependent. For a path $Q$ in $G$ of 
length $d(Q)$, there is a set $P_Q$ of $O(1/\epsilon)$ portals such that for any 
vertices $u$ and $v$ shortest-separated by $Q$, 
$\min_{p\in P_Q} d_G(u,p)+d_G(p,v)\leq d_G(u,v)+\epsilon d(Q)$ \cite{KS98}. Based 
on this and a scaling technique, Kawarabayashi et al. \cite{KST13} give another 
$(1+\epsilon)$-approximate distance oracle: Create subgraphs of $G$ such that 
the vertices in each subgraph satisfy certain distance property (scaling). Each 
subgraph $H$ of $G$ is decomposed by shortest paths into a $\rho$-division of $H$ 
which consists of $O(|V(H)|/\rho)$ subgraphs of $H$, each has size $O(\rho)$. For 
each subgraph $X$ of $H$, let $\BB(X)$ be the set of shortest paths used to
separate $X$ from the rest of $H$. For each path $Q\in \BB(X)$, a portal set 
$P_Q$ is selected. For vertices $u$ and $v$ shortest-separated by some path 
$Q\in \BB(X)$, $\min_{p\in P_Q,Q\in \BB(X)} d_H(u,p)+d_H(p,v)$
is used to approximate $d_G(u,v)$. This oracle does not keep the distances 
$d_H(u,p)$ and $d_H(p,v)$ but uses the distance oracle in \cite{MS12} to get the 
distances. By choosing an appropriate value $\rho$, the oracle has a better 
product of query time and oracle size than that of Thorup's oracle.

We also use the scaling technique to create subgraphs of $G$. We decompose each 
subgraph $H$ of $G$ into a balanced recursive subdivision as in Thorup's oracle. 
For each subgraph $X$ of $H$ and each shortest path $Q$ used to decompose $X$, we 
choose one set $P_Q$ of $O(1/\epsilon)$ portals on $Q$ for all vertices in $X$. 
A new ingredient in our oracle is to use a more time efficient data structure
to approximate $d_G(u,v)$ instead of 
$\min_{p\in P_Q,Q\in \PP(X)} d_H(u,p)+d_H(p,v)$. Using an approach in \cite{WY13}, 
we show that the vertices in $V(X)$ can be partitioned into $s=f(\epsilon)$ 
classes $A_1,...,A_s$ such that for every two classes $A_i$ and $A_j$, there is a 
key portal $p_{ij}\in P_Q$ and for any $u\in A_i$ and $v\in A_j$, if $u$ and $v$ 
are shortest-separated by $Q$ then 
$d_H(u,p_{ij})+d_H(p_{ij},v)\leq (1+\epsilon)d_G(u,v)$ and 
$d_H(u,p_{ij})+d_H(p_{ij},v)$ can be computed in $O(1)$ time. This gives a 
$(1+\epsilon)$-approximate distance oracle with $O(1)$ query time.

Our computational model is word RAM, which models what we can program using 
standard programming languages such as C/C++. In this model, a word is assumed 
big enough to store any vertex identifier or distance. We also assume basic
operations, which include addition, subtraction, multiplication, bitwise 
operations (AND, OR, NEGATION) and left/right cyclic shift on a word have unit 
time cost.

The rest of the paper is organized as follows. In the next section, we give
preliminaries of the paper and review the techniques on which our oracles build.
In Section 3, we present distance oracles with additive stretch. In Section 4,
we give the $(1+\epsilon)$-approximate distance oracles which use the additive
stretch oracles as subroutines. The final section concludes the paper.

\section{Preliminaries}

An undirected graph $G$ consists of a set $V(G)$ of vertices and a set $E(G)$ of 
edges. For a subset $A\subseteq E(G)$, we denote by $V(A)$ the set of vertices 
incident to at least one edge of $A$. For 
$A\subseteq E(G)$ and $W\subseteq V(G)$, we denote by $G[A]$ and $G[W]$ the 
subgraphs of $G$ induced by $A$ and $W$, respectively. A graph $H$ is a subgraph 
of $G$ if $V(H)\subseteq V(G)$ and $E(H)\subseteq E(G)$. 

A path between vertices $u$ and $v$ in $G$ is a sequence of edges $e_1,..,e_k$, 
where $e_i=\{v_{i-1},v_i\}$ for $1\leq i\leq k$, $u=v_0$, $v=v_k$, and the 
vertices $v_0,...,v_k$ are distinct. For any edge $e$, let $l(e)$ be the length 
of $e$. The length of path $Q=e_1,...,e_k$ is $d(Q)=\sum_{1\leq i\leq k} l(e_i)$. 
A path $Q$ is a shortest path between vertices $u$ and $v$ if $d(Q)$ is the 
minimum among those of all paths between $u$ and $v$. The distance between 
vertices $u$ and $v$ in $G$, denoted by $d_G(u,v)$, is the length of a shortest 
path between $u$ and $v$. For each vertex $u$ in $G$, the {\em eccentricity} of 
$u$ is $\lambda(u)=\max_{v\in V(G)} d_G(u,v)$. The {\em radius} of $G$ is 
$r(G)=\min_{u\in V(G)} \lambda(u)$. The {\em diameter} of $G$ 
is $d(G)=\max_{u\in V(G)} \lambda(u)$.

A graph is planar if it has a planar embedding (a drawing on a sphere without edge 
crossing). In the rest of this paper, graphs are undirected planar graphs with
non-negative edge lengths unless otherwise stated. 

A basic approach in this paper is to decompose graph $G$ into subgraphs by 
shortest paths. A set $\PP$ of shortest paths in graph $G$ is a {\em shortest path 
separator} of $G$ if $G[V(G)\backslash W]$, $W=\cup_{Q\in \PP} V(Q)$, has at least 
$t\geq 2$ connected nonempty subgraphs $G_1,..,G_t$ of $G$. A set $\QQ$ of paths 
{\em separates} subgraphs $G_i$ and $G_j$, $i\neq j$, if for any vertex $u$ 
in $G_i$ and any vertex $v$ in $G_j$, any path between $u$ and $v$ intersects a 
path in $\QQ$. For a subgraph $G_i$ of $G$, a set $\BB(G_i)$ of paths is 
a {\em boundary} of $G_i$ if $\BB(G_i)$ separates $G_i$ and the rest of $G$ 
and for every path $Q\in \BB(G_i)$, there is an edge connecting $Q$ and $G_i$.
For $\alpha>0$, a shortest path separator $\PP$ of $G$ is called 
{\em $\alpha$-balanced} if $|V(G_i)|\leq \alpha|V(G)|$ holds for every subgraph 
$G_i$. An $\alpha$-{\em balanced recursive subdivision} of $G$ is a structure that 
$G$ is decomposed into subgraphs $G_1,..,G_t$ by an $\alpha$-balanced separator 
and for each $1\leq i\leq t$, $G_i$ is decomposed recursively until each subgraph 
is reduced to a pre-defined size. In the recursive decomposition of $G_i$, the
subset of the shortest path separator $\PP$ of $G$ that forms a boundary $\BB(G_i)$ 
of $G_i$ is included in computing a shortest path separator of $G_i$. Let $T_r$ be 
a shortest path spanning tree of graph $G$ rooted at a vertex $r$. Every path in 
$T_r$ from the root $r$ to any vertex is a shorest path and called a 
{\em root path}. We use Thorup's method \cite{Tho04} to compute a 
$\frac{1}{2}$-balanced recursive subdivision using shortest path separators 
composed of root paths in $T_r$ (based on the result in \cite{LT79}, this can
be done in linear time).

We now briefly describe Thorup's method. Readers may refer to Section $2.5$ in 
\cite{Tho04} for more details. A recursive subdivision of $G$ can be viewed as a 
rooted tree $T_G$ with each vertex of $T_G$ (called a {\em node}, to be 
distinguished from a vertex of $G$) representing a subgraph of $G$ and the root 
node representing $G$. Each node in $T_G$ with node degree one is called a 
{\em leaf node}, otherwise an {\em internal node}. We identify subgraphs with 
their corresponding nodes in $T_G$ when convenient. The root node of $T_G$ has 
depth $0$ and for any node $X$ of $T_G$, the depth of $X$ is the number of edges 
of $T_G$ from $X$ to the root node. The depth of $T_G$ is the largest depth of 
any node in $T_G$. For each node $X$ of $T_G$, let $\BB(X)$ be the set of 
root paths that forms a boundary of $X$ ($\BB(G)=\emptyset$). Let $X\cup \BB(X)$ 
denote the subgraph of $G$ induced by $V(X)\cup V(\BB(X))$. Let $X+\BB(X)$ denote 
the graph obtained by removing some vertices from  $X\cup \BB(X)$ as follows: for 
every vertex $v$ of $\BB(X)$ that has degree two in $X\cup \BB(X)$, its incident 
edges $(u,v)$ and $(v,w)$ are replaced by edge $(u,w)$ whose length is the sum of 
the length of $(u,v)$ and that of $(v,w)$. For each internal node $X$, a 
$\frac{1}{2}$-balanced shortest path separator $\PP(X)$ of root paths is 
used to decompose $X$ into subgraphs $X_1,..,X_t$, $t\geq 2$, as follows: Let 
$W=V(\PP(X))$ and $X^*_1,..,X^*_t$ be the connected components of 
$G[V(X+\BB(X))\setminus W]$. 
Then $E(X_i)=E(X)\cap E(X^*_i)$, $1\leq i\leq t$. Note that $\PP(X)$ separates 
$X_i$ from $X_j$ in $X$ and $\BB(X)\cup \PP(X)$ separates $X_i$ from 
$X_j$ in $G$ for $1\leq i,j\leq t$ and $i\neq j$. We now state some important 
properties of the $\frac{1}{2}$-balanced recursive subdivision in the next Lemma. 

\begin{lemma} \cite{Tho04}
Given a graph $G$ and a shortest path spanning tree $T_r$ of $G$, a 
$\frac{1}{2}$-balanced recursive subdivision $T_G$ of $G$ can be computed in 
$O(n\log n)$ time such that for each internal node $X$ of $T_G$,
$|V(X_i)|\leq |V(X)|/2$ ($1\leq i\leq t$) and $|\PP(X)|=O(1)$, and for each
node $X$, $|\BB(X)|=O(1)$. Moreover, for each node $X$ of $T_G$ and each root 
path $Q$ of $T_r$, if $Q\in\BB(X)$, then $Q\in\PP(X^\prime)$ for 
some ancestor $X^\prime$ of $X$ in $T_G$.
\label{lem-subdivision}
\end{lemma}
The recursive subdivision of $G$ in Lemma~\ref{lem-subdivision} will be used in 
our oracles. Note that since the size of a subgraph is reduced by at least $1/2$, 
the depth of $T_G$ is bounded above by $\log n$. For every vertex $v\in V(G)$, we 
define the {\it home} of $v$, denoted by $X_v$, to be the node of $T_G$ of largest 
depth that contains $v$. For any $u,v\in V(G)$, we define $X_{u,v}$ to be the 
{\it nearest common ancestor} of $X_u$ and $X_v$ in $T_G$. Harel and Tarjan show 
in \cite{HT84} that after a linear time preprocessing, the nearest common ancestor 
of any two nodes in a tree can be found in $O(1)$ time. 

Let $Q$ be a shortest path in $G$ and $\epsilon>0$. Thorup shows that for every 
vertex $u$ in $G$, there is subset $P_Q(u)\subseteq V(Q)$ of $O(1/\epsilon)$ 
vertices such that for any vertices $u$ and $v$ shortest-separated by $Q$
\[
d_G(u,v)\le \min_{p\in P_Q(u),q\in P_Q(v)}
d_G(v,p)+d_G(p,q)+d_G(q,v)\le (1+\epsilon)d_G(u,v).
\]
The vertices of $P_Q(u)$ are called {\em portals} on $Q$ for $u$. For every 
subgraph $X$ in a $\frac{1}{2}$-balanced recursive subdivision of $G$ and every 
shortest path $Q\in \BB(X)\cup\PP(X)$, by keeping the distance from 
each vertex $u$ in $X$ to every portal in $P_Q(u)$ explicitly, Thorup shows the 
following result.  
\begin{lemma} \cite{Tho04}
For graph $G$ and $\epsilon>0$, there is a $(1+\epsilon)$-approximate distance 
oracle with $(1/\epsilon)$ query time, $O(n\log n/\epsilon)$ size and 
$O(n\log^3 n/\epsilon^2)$ pre-processing time. Especially for $\epsilon=1$, 
there is a $2$-approximate distance oracle for $G$ with $O(1)$ query time, 
$O(n\log n)$ size and $O(n\log^3 n)$ pre-processing time.
\label{lem-estimate}
\end{lemma} 
Our oracles will use this oracle for $\epsilon=1$ (any constant works) to get a 
rough estimation of $d_G(u,v)$. 

To reduce the query time to a constant independent of $\epsilon$, we will use a 
portal set $P_Q$ independent of vertex $u$. For vertices $u$ and $v$ 
shortest-separated by a path $Q$, $d_G(u,v)=\min_{p\in V(Q)} d_G(u,p)+d_G(p,v)$.
For a $P_Q\subseteq V(Q)$, $\min_{p\in P_Q} d_G(u,p)+d_G(p,v)$ approximates
$d_G(u,v)$. The following result will be used.
\begin{lemma} \cite{KS98}
For a path $Q$ in $G$, $\epsilon>0$ and $D\geq d(Q)$, a set $P_Q$ of 
$O(1/\epsilon)$ vertices in $V(Q)$ can be selected in $O(|V(Q)|)$ time such that 
for any pair of vertices $u$ and $v$ shortest-separated by $Q$, 
$d_G(u,v)\leq \min_{p\in P_Q} d_G(u,p)+d_G(p,v)\leq d_G(u,v)+\epsilon D$. 
\label{lem-portal}
\end{lemma}
The set $P_Q$ in Lemma~\ref{lem-portal} is called the {\em $\epsilon$-portal set 
(with respect to $D$)} and every vertex in $P_Q$ is called a {\it portal}. Given 
a path $Q$ starting from a vertex $r$, $\epsilon>0$ and $D\geq d(Q)$, $P_Q$ can 
be computed as follows: add $r$ to $P_Q$, traverse along $Q$ from $r$ and add a 
vertex $v\in V(Q)$ to $P_Q$ if $d_G(u,v)\geq\epsilon D/2$, where $u$ is the last 
added portal in $P_Q$. To apply the $\epsilon$-portal set to our oracle, we 
further need to guarantee $d_G(u,v)=\Omega(D)$ for vertices $u$ and $v$ in 
question. We will use the {\em sparse neighborhood covers} introduced in 
\cite{ABCP98,AP90,BLT07} of $G$ to achieve this goal.
\begin{lemma} \cite{BLT07}
For $G$ and $\gamma\geq 1$, connected subgraphs 
$G(\gamma,1),\dots,G(\gamma,n_{\gamma})$ of 
$G$ with the following properties can be computed in $O(n\log n)$ time:
\begin{enumerate}
\item For each vertex $u$ in $G$, there is at least one $G(\gamma,i)$ that 
contains $u$ and every $v$ with $d_G(u,v)\leq \gamma$.
\item Each vertex $u$ in $G$ is contained in at most 18 subgraphs. 
\item Each subgraph $G(\gamma,i)$ has radius $r(G(\gamma,i))\leq 24\gamma-8$. 
\end{enumerate}
\label{lem-cover}
\end{lemma}

\section{Oracle with additive stretch}

We first give a distance oracle which for any vertices $u$ and $v$ in $G$, and 
any $\epsilon_0>0$, returns $\tilde{d}(u,v)$ with 
$d_G(u,v)\leq \tilde{d}(u,v)\leq d_G(u,v)+7\epsilon_0 d(G)$. Based on the scaling 
technique in \cite{KST13} and Lemma~\ref{lem-cover}, this oracle will be extended 
to an oracle stated in Theorem~\ref{theo-1} for $G$ in the next section.

We start with a basic data structure which keeps the following information:
\begin{itemize}
\item A $\frac{1}{2}$-balanced recursive subdivision $T_G$ 
of $G$ as in Lemma~\ref{lem-subdivision}, each leaf node in $T_G$ has size 
$O(2^{(1/\epsilon_0)})$.
\item A table storing $X_v$ for every $v\in V(G)$.
\item A data structure with $O(1)$ query time to find the nearest common ancestor 
$X_{u,v}$ of $X_u$ and $X_v$ in $T_G$ for any $u$ and $v$ in $G$.
\item For each internal node $X$ of $T_G$, an $\epsilon_0$-portal set $P_Q$ for 
every shortest path $Q\in \PP(X)\cup \BB(X)$. For every $P_Q$, 
every $u\in V(X)$ and every portal $p\in P_Q$, distance $\hat{d}(u,p)$ with 
\[
d_G(u,p)\leq \hat{d}(u,p)\leq d_G(u,p)+\epsilon_0 d(G).
\] 
\item For every leaf node $X$ and every pair of $u$ and $v$ in $X$, we keep
\[
\tilde{d}(u,v)=\min\{d_X(u,v), \min_{p\in P_Q,Q\in \BB(X)} 
\hat{d}(u,p)+\hat{d}(p,v)\}.
\]
\end{itemize}
The data structure above gives a distance oracle with $3\epsilon_0 d(G)$ additive 
stretch and $O(1/\epsilon_0)$ query time: Given vertices $u$ and $v$, if $X_{u,v}$ 
is a leaf node then $\tilde{d}(u,v)$ can be found in $O(1)$ time. Otherwise, $u$ 
and $v$ must be shortest-separated by some path in 
$\BB(X_{u,v})\cup\PP(X_{u,v})$.
Let
\[\tilde{d}(u,v)=\min_{p\in P_Q,Q\in \BB(X_{u,v})\cup \PP(X_{u,v})} 
\hat{d}(u,p)+\hat{d}(p,v)\] and 
\[
q=\arg_{p\in P_Q,Q\in \BB(X_{u,v})\cup \PP(X_{u,v})} 
\min\{d_G(u,p)+d_G(p,v)\}.
\] 
From $\hat{d}(u,q)\leq d_G(u,q)+\epsilon_0 d(G)$,
$\hat{d}(q,v)\leq d_G(q,v)+\epsilon_0 d(G)$ and Lemma~\ref{lem-portal},
\begin{eqnarray*}
d_G(u,v) &\leq & \tilde{d}(u,v) \leq \hat{d}(u,q)+\hat{d}(q,v)\\
&\leq & d_G(u,q)+d_G(q,v)+2\epsilon_0 d(G)\leq
d_G(u,v)+3\epsilon_0 d(G).
\end{eqnarray*}
$\tilde{d}(u,v)$ can be computed in $O(1/\epsilon_0)$ time because 
$|P_Q|=O(1/\epsilon_0)$ and $|\BB(X_{u,v})\cup \PP(X_{u,v})|=O(1)$. 

We first reduce the query time for internal nodes in the above oracle to a 
constant independent of $\epsilon_0$ and then analyse the preprocessing time of 
the distance oracle. For $z>0$, let $f(z)=2^{O(1/z)}$. Based on an approach in 
\cite{WY13}, we show that for each internal node $X$ and each path 
$Q\in \BB(X)\cup \PP(X)$, the vertices in $V(X)$ can be partitioned 
into $f(\epsilon_0)$ classes such that for any two classes $A_i$ and $A_j$, there 
is a key portal $p_{ij}\in P_Q$ and for every $u\in A_i$ and every $v\in A_j$ 
shortest-separated by $Q$, 
$\hat{d}(u,p_{ij})+\hat{d}(p_{ij},v)\leq d_G(u,v)+7\epsilon_0 d(G)$.
By keeping the classes and key portals, the query time is reduced to $O(1)$. We 
first define the classes.
\begin{definition}
Let $Q$ be a shortest path in $G$, $r(G)\leq D\leq d(G)$ and $P_Q=\{p_1...,p_l\}$ 
be an $\epsilon_0$-portal set (with respect to $D$) on $Q$. The vertices of $G$ 
are partitioned into classes based on $\hat{d}(u,p_i), p_i\in P_Q$ as follows. 
For each vertex $u$, a vector $\vec{\Gamma}_u=(a_1,...,a_l)$ is defined such 
that for $1\leq i\leq l$, $a_i=\ceil{\hat{d}(u,p_i)/(\epsilon_0D)}$. Vertices $u$ 
and $v$ are in the same class if and only if $\vec{\Gamma}_u=\vec{\Gamma}_v$. 
\label{def-class}
\end{definition}
The following property of the classes defined above is straightforward.
\begin{property}
Let $Q$ be a shortest path in $G$,  $r(G)\leq D\leq d(G)$ and $P_Q$ be an
$\epsilon_0$-portal set with respect to $D$ on $Q$. Let $A$ be any class of 
vertices in $G$ defined in Definition~\ref{def-class}. For any two vertices 
$u,v\in A$ and any portal $p\in P_Q$, 
$\hat{d}(u,p)-\epsilon_0D\leq \hat{d}(v,p)\leq \hat{d}(u,p)+\epsilon_0D.$
\label{prop-1}
\end{property}
We show more properties of the classes defined above in the next two lemmas. 
\begin{lemma}
Let $Q$ be a shortest path in $G$,  $r(G)\leq D\leq d(G)$ and $P_Q$ be an 
$\epsilon_0$-portal set with respect to $D$ on $Q$. Let $A_i$ and $A_j$ be any 
two classes of vertices in $G$ defined in Definition~\ref{def-class}. There is a 
key portal $p_{ij}\in P_Q$ such that for any vertices $u\in A_i$ and $v\in A_j$ 
shortest-separated by $Q$, 
$d_G(u,v)\leq \hat{d}(u,p_{ij})+\hat{d}(p_{ij},v)\leq d_G(u,v)+7\epsilon_0 d(G)$.
\label{lem-keyportal}
\end{lemma}
\begin{proof}
We choose arbitrarily a vertex $x\in A_i$ and a vertex $y\in A_j$. 
Let $p_{ij}=\arg_{p_i\in P_Q} \min\{\hat{d}(x,p_i)+\hat{d}(p_i,y)\}$ be the key
portal. For any $u\in A_i$ and $v\in A_j$ shortest-separated by $Q$, let
$q=\arg_{p_i\in P_Q} \min\{d_G(u,p_i)+d_G(p_i,v)\}$ and let 
$p=\arg_{p_i\in P_Q} \min\{\hat{d}(u,p_i)+\hat{d}(p_i,v)\}$. Then
\begin{eqnarray*}
\hat{d}(u,p)+\hat{d}(p,v)&\leq& \hat{d}(u,q)+\hat{d}(q,v)\\
&\leq& d_G(u,q)+d_G(q,v)+2\epsilon_0 d(G)\leq d_G(u,v)+3\epsilon_0 d(G),
\end{eqnarray*}
because $\hat{d}(u,q)\leq d_G(u,q)+\epsilon_0 d(G)$, 
$\hat{d}(q,v)\leq d_G(q,v)+\epsilon_0 d(G)$, $P_Q$ is an $\epsilon_0$-portal 
set and Lemma~\ref{lem-portal}. From $u,x\in A_i$, Property~\ref{prop-1} and 
$D\leq d(G)$,
\[
\hat{d}(u,p_i)\leq\hat{d}(x,p_i)+\epsilon_0D\leq\hat{d}(x,p_i)+
\epsilon_0d(G)\leq\hat{d}(u,p_i)+2\epsilon_0d(G)
\] 
for every $p_i\in P_Q$. The same relations hold for $v,y$ because they are in 
$A_j$. So
\begin{eqnarray*}
\hat{d}(u,p_{ij})+\hat{d}(p_{ij},v) &\leq & \hat{d}(x,p_{ij})+\hat{d}(p_{ij},y)
+ 2\epsilon_0 d(G) \\
&\leq & \hat{d}(x,p)+\hat{d}(p,y)+2\epsilon_0 d(G) \leq 
\hat{d}(u,p)+\hat{d}(p,v)+4\epsilon_0 d(G).
\end{eqnarray*}
Therefore,
\begin{eqnarray*}
d_G(u,v) &\leq & \hat{d}(u,p_{ij})+\hat{d}(p_{ij},v)\leq 
\hat{d}(u,p)+\hat{d}(p,v)+4\epsilon_0 d(G) \\
 & \leq & d_G(u,v)+7\epsilon_0 d(G).
\end{eqnarray*}
This completes the proof of the lemma.
$\Box$
\end{proof}

\begin{lemma}
The total number of classes by Definition~\ref{def-class} is $f(\epsilon_0)$.
\label{lem-class}
\end{lemma}

\begin{proof} Essentially, this result is proved by Weimann and 
Yuster in \cite{WY13} but somehow hidden in other details.
Below we give a self-contained proof of the lemma.
For each vector 
$\vec{\Gamma}_u=(a_1,..,a_l)$, let
$\vec{\Gamma}^*_u=(a_1,(a_2-a_1),(a_3-a_2),..,(a_l-a_{l-1}))$. Then 
$\vec{\Gamma}_u=\vec{\Gamma}_v$ if and only if 
$\vec{\Gamma}^*_u=\vec{\Gamma}^*_v$. So we just need to prove that the total 
number of different $\vec{\Gamma}^*_u$ is $f(\epsilon_0)$. From Definition 
\ref{def-class},
\begin{eqnarray*}
    \abs{a_i-a_{i-1}} &=&\abs{\ceil{\frac{\hat{d}(u,p_i)}{\epsilon_0D}}-
	\ceil{\frac{\hat{d}(u,p_{i-1})}{\epsilon_0D}}}\\
    &\leq & \abs{\frac{\hat{d}(u,p_i)-\hat{d}(u,p_{i-1})}{\epsilon_0D}}+1\\
    &\leq& \abs{\frac{d_G(u,p_i)-d_G(u,p_{i-1})}{\epsilon_0D}}+2\leq\frac{d_G(p_{i-1},p_i)}{\epsilon_0D}+2.
\end{eqnarray*}
Since $P_Q$ is an $\epsilon_0$-portal set, $l=O(1/\epsilon_0)$. 
So 
\[
\sum_{2\leq i\leq l} \abs{a_i-a_{i-1}}\leq\frac{d_G(p_1,p_l)}{\epsilon_0D}+2l
=O(1/\epsilon_0).
\] 
Therefore there are
$2^{O(1/\epsilon_0)}$ different vectors of
$(a_1,\abs{a_2-a_1},\abs{a_3-a_2},..,\abs{a_l-a_{l-1}})$. The $i$'th element of
$(a_1,(a_2-a_1),(a_3-a_2),..,(a_l-a_{l-1}))$ is either $\abs{a_i-a_{i-1}}$ or
$-\abs{a_i-a_{i-1}}$. Therefore, there are $2^{O(1/\epsilon_0)}$ different 
$\vec{\Gamma}^*_u$.
$\Box$
\end{proof} 

Notice that we can assume that for each internal node $X$, the number of classes
Definition~\ref{def-class} is at most $|V(X)|^2$ because otherwise, instead of
partitioning the vertices into classes, we can simply use a $|V(X)|\times |V(X)|$ 
distance array to keep the shortest distance between every pair of vertices in $X$.

Now we are ready to show a data structure $\DS_0$ for our oracle with 
$7\epsilon_0d(G)$ additive stretch. $\DS_0$ contains the basic data structure 
given above and the following additional information:
\begin{itemize}
\item For each internal node $X$ of $T_G$ and each shortest path 
$Q\in \BB(X)\cup \PP(X)$, let $A_1^Q,...,A_s^Q$ be the classes of 
vertices in $V(X)$ defined in Definition \ref{def-class}. For each vertex 
$u\in V(X)$, we give an index $I^Q_X(u)$ with $I^Q_X(u)=i$ if $u\in A^Q_i$; 
and an $s\times s$ array $C_Q$ with $C_Q[i,j]$ containing the key portal 
$p^Q_{ij}$ for classes $A^Q_i$ and $A^Q_j$. 
\end{itemize}

We now describe how to compute the distances $\hat{d}(d,p)$ for internal nodes as 
defined in the basic data structure. The method is essentially the same as in the 
fast construction in \cite{Tho04}, but simpler as the portal sets we use are not 
vertex dependent. We use Lemma~\ref{lem-subdivision} to get the recursive 
subdivision $T_G$ of $G$. Let $T_r$ be the shortest path spanning tree of $G$ as 
defined in Section $2$ and let $D$ be the largest length of any root path of $T_r$. 
By a depth first search of $T_r$ from $r$, 
we compute an $\epsilon_0$-portal set $P_Q$ (with respect to $D$) and an 
{\em auxiliary $(\epsilon_0/\log n)$-portal set} $\Gamma_Q$ (with respect to $D$) 
for every $Q\in\PP(X)\cup\BB(X)$, $X\in V(T_G)$. We compute the distances 
$\hat{d}(u,p)$ for every internal node $X$ in a top-down traversal on $T_G$ from 
root $G$. We use Dijkstra's algorithm to compute $\hat{d}(u,p)$ for every $u$ in 
$X$ and every $p\in P_Q\cup \Gamma_Q$, $Q\in \PP(X)$. 
Let $X\star\BB(X)$ denote the graph obtained from adding to 
$X\cup\BB(X)$ the edges $\{u,p'\}$ with length $\hat{d}(u,p')$ for every 
$u$ in $X$ and every $p'\in P_{Q'}\cup \Gamma_{Q'}$, where $Q^\prime\in\BB(X)$
and $Q^\prime\in \PP(X')$ for some ancestor $X'$ of $X$, and 
then removing degree two vertices of $\BB(X)$ as what we do for 
$X+\BB(X)$. For the root node, the computation is on $G$. For an internal 
node $X\neq G$, the computation is on $X\star\BB(X)$. Note that 
$X\star \BB(X)$ may not be planar and since 
$|\BB(X)|=O(1)$, $|V(X\star\BB(X))|$ is linear in the number of 
edges of $G$ incident to vertices of $X$ plus the number of portals in each path.
Notice that $Q'$ is in $\PP(X')$ for some internal node $X'$ which is an 
ancestor of $X$ in $T_G$. So the distances $\hat{d}(u,p')$ have been computed when 
we construct $X\star \BB(X)$.
Note that for some vertex $u$ and portal $p$, $\hat{d}(u,p)$ may
be computed multiple times. But the value of $\hat{d}(u,p)$ does not change:
let $H_i, 1\leq i$, be the graph on which $\hat{d}(u,p)$ is computed for the
$i$th time; $\hat{d}(u,p)$ does not increase because the edge $\{u,p\}$ is 
contained in $H_i$ for $i\geq 2$; and $\hat{d}(u,p)$ does not decrease because 
$H_{i+1}$ is a subgraph of $H_i$ for $i\geq 2$.
For $X$ with $|V(X)|\geq \log n/\epsilon_0$, we run Dijkstra's algorithm using 
every $p\in P_Q\cup \Gamma_Q$ as the source. For $X$ with 
$|V(X)|<\log n/\epsilon_0$, we run Dijkstra's algorithm using every $u$ in $X$ 
as the source. After the distances $\hat{d}(u,p)$, $p \in P_Q\cup \Gamma_Q$, 
for all internal nodes have been computed, we only keep the distances 
$\hat{d}(u,p)$ to the portals $p\in P_Q$ for every internal node.

In the next lemma, we show that the distances $\hat{d}(u,p)$ computed above 
meet the requirement of $\DS_0$. 
\begin{lemma}
For every internal node $X$ of $T_G$, every vertex $u$ in $X$ and every portal 
$p\in P_Q\cup \Gamma_Q$, $Q\in \PP(X)$, 
$\hat{d}(u,p)\leq d_G(u,p)+\epsilon_0d(G)$.
\label{lem-distance}
\end{lemma}
\begin{proof}
For every internal node $X$ of depth $k$ in $T_G$, every vertex $u$ in $X$ and
every portal $p\in P_Q\cup \Gamma_Q, Q\in \PP(X)$, we prove by induction that
$\hat{d}(u,p)\leq d_G(u,p)+\frac{k\epsilon_0}{\log n}d(G)$. For the root node 
(of depth 0), $\hat{d}(u,p)=d_G(u,p)$ because the distances are computed on $G$. 
Assume that for every internal node of depth at most $k-1\geq 0$, 
$\hat{d}(u,p)\leq d_G(u,p)+\frac{(k-1)\epsilon_0}{\log n} d(G)$. Let $X$ be a node 
of depth $k$. For $u$ in $X$ and $p\in P_Q\cup \Gamma_Q$, $Q\in \PP(X)$, let 
$P(u,p)$ be a shortest path between $u$ and $p$. If $P(u,p)$ contains only edges in 
$X$ then $\hat{d}(u,p)=d_G(u,p)$ and the statement is proved. Otherwise, $P(u,p)$ 
can be partitioned into two subpaths $P(u,y)$ and $P(y,p)$, where every vertex of
$P(u,y)$ except $y$ is in $X$ and $y$ is a vertex of a path $Q'\in \BB(X)$. 
Note that $y$ is incident to some vertex of $X$ so $y$ appears in 
$X\star\BB(X)$. From the way $\Gamma_{Q'}$ is computed and the fact that 
$D\leq d(G)$, where $D$ is used for computing $\Gamma_{Q'}$, there is a portal 
$p_y\in \Gamma_{Q'}$ such that 
$d_{Q'}(y,p_y)=d_G(y,p_y)\leq \frac{\epsilon_0}{2\log n}d(G)$. Let $X^{\prime}$ be 
an ancestor of $X$ such that $Q'\in \PP(X')$.
Note that $\hat{d}(u,p_y)$ is computed in $X'\star \BB(X^\prime)$ and that $y$, 
$p_y$ and $X$ (and therefore $P(u,y)$) are all contained in 
$X^{\prime}\star\BB(X^\prime)$. Therefore,
\[
\hat{d}(u,p_y)\leq d(P(u,y))+d_{Q'}(y,p_y)\leq 
d(P(u,y))+\frac{\epsilon_0}{2\log n}d(G)
\]
and
\[
d_G(p_y,p)\leq d(P(y,p))+d_G(y,p_y)\leq d(P(y,p))+\frac{\epsilon_0}{2\log n}d(G).
\]
The distance $\hat{d}(p_y,p)$ has been computed in a node $X'$ which is an 
ancestor of $X$ and has depth at most $k-1$. So
$\hat{d}(p_y,p) \leq d_G(p_y,p)+\frac{(k-1)\epsilon_0}{\log n}d(G)$.
Because edges $\{u,p_y\}$ and $\{p_y,p\}$ with lengths $\hat{d}(u,p_y)$ and 
$\hat{d}(p_y,p)$ are contained in the graph $X\star\BB(X)$,
\begin{eqnarray*}
\hat{d}(u,p) & \leq & \hat{d}(u,p_y)+\hat{d}(p_y,p) \\
& \leq & d(P(u,y))+\frac{\epsilon_0}{2\log n}+d(P(y,p))+
\frac{\epsilon_0}{2\log n} + \frac{(k-1)\epsilon_0}{\log n}d(G) \\
 & = & d_G(u,p)+ \frac{k\epsilon_0}{\log n}d(G).
\end{eqnarray*}
Since each node in $T_G$ has depth at most $\log n$, 
$\hat{d}(u,p)\leq \epsilon_0 d(G)$.
$\Box$
\end{proof}

The next three lemmas give the pre-procesing time, space requirement and query
time for data structure $\DS_0$. 
\begin{lemma}
For graph $G$ and $\epsilon_0>0$, data structure $\DS_0$ can be computed in
$O(n(\log^3 n/\epsilon_0^2+f(\epsilon_0)))$ time.
\label{lem-processing}
\end{lemma}
\begin{proof}
Let $T_G$ be the recursive subdivision of $G$ and $b=2^{(1/\epsilon_0)}$. 
It takes $O(n\log n)$ time to compute $T_G$ (Lemma~\ref{lem-subdivision}). 
It takes $O(n)$ time to compute the data structure that can answer the least 
common ancestor of any two nodes in $T_G$ in $O(1)$ time \cite{HT84}, $O(n\log n)$
time to compute $X_v$ for every $v\in V(G)$, and $O(n)$ time to compute 
$P_Q\cup \Gamma_Q$ for every path $Q\in\BB(X)\cup\PP(X)$, $X\in V(T_G)$. 

For every node $X$, let $M(X)$ be the number of edges in $G$ incident to vertices 
in $X$. Then $\sum_{X\in T_G}M(X)=O(n\log n)$. For every path $Q$, 
$|P_Q\cup \Gamma_Q|=O(\log n/\epsilon_0)$ and for every node $X$ in $T_G$, 
$|\BB(X)\cup \PP(X)|=O(1)$. From this, for each internal node $X$, 
$X\star\BB(X)$ has $O(M(X)+\log n/\epsilon_0)$ vertices and 
$O(M(X)+\log n/\epsilon_0)$ edges. Dijkstra's algorithm is executed 
$\min\{M(X),\log n/\epsilon_0\}$ times for each node $X$. It takes 
$O(M(X)(\log n/\epsilon_0)^2)$ time to compute all $\hat{d}(u,p)$ for 
node $X$. Since the sum of $M(X)$ for all nodes $X$ of the same depth is $O(n)$, 
it takes $O(n(\log n/\epsilon_0)^2)$ time for all internal nodes of the same 
depth. Since $T_G$ has depth $O(\log n)$, it takes $O(n\log^3 n/\epsilon_0^2)$ 
time to compute all distances $\hat{d}(u,p)$ for all internal nodes.

To find $\tilde{d}(u,v)$ for a leaf node $X$, we first use Dijkstra's algorithm
to compute $d_X(u,v)$, taking every vertex of $X$ as the source. This takes
$O(b^2\log b)=O(b^2/\epsilon_0)$ time for one leaf node since $|V(X)|=O(b)$ and 
$O(nb/\epsilon_0)$ time for all leaf nodes since the sum of $|V(X)|$ for all
leaf nodes $X$ is $O(n)$. Then we compute 
\[
\tilde{d}(u,v)=\min\{d_X(u,v), \min_{p\in P_Q,Q\in \BB(X)} 
\hat{d}(u,p)+\hat{d}(p,v)\}.
\]
From $|P_Q|=O(1/\epsilon_0)$ for $Q\in \BB(X)$ and $|\BB(X)|=O(1)$,
this takes $O(b^2/\epsilon_0)$ time for one leaf node and $O(nb/\epsilon_0)$ time 
for all leaf nodes. The total time to compute $\tilde{d}(u,v)$ for all leaf nodes
is $O(nb/\epsilon_0))=O(nf(\epsilon_0))$.

The value $D$ for computing the classes can be found in $O(n)$ time.
Since there are $O(n)$ internal nodes, by Lemma~\ref{lem-class}, it takes
$O(nf(\epsilon_0)(1/\epsilon_0))=O(nf(\epsilon_0))$ time to compute all classes
and key portals. Therefore,
$\DS_0$ can be computed in $O(n(\log^3 n/\epsilon_0^2+f(\epsilon_0)))$ time.
$\Box$
\end{proof}

\begin{lemma}
For graph $G$ and $\epsilon_0>0$, the space requirement for data structure 
$\DS_0$ is $O(n(\log n/\epsilon_0+f(\epsilon_0)))$.
\label{lem-size}
\end{lemma}
\begin{proof}
Let $T_G$ be the recursive subdivision of $G$ in $\DS_0$ and 
$b=2^{(1/\epsilon_0)}$. 
Each leaf node $X$ has $O(b)$ vertices and requires $O(b^2)$ space to keep the
distances $\tilde{d}(u,v)$ for $u,v$ in the node. From this and the sum of 
$|V(X)|$ for all leaf nodes is $O(n)$, the space for all leaf nodes is 
$O(nb)=O(nf(\epsilon_0))$. By Lemma~\ref{lem-subdivision}, the sum of $|V(X)|$ 
for all nodes $X$ in $T_G$ is $O(n\log n)$. From
$|\BB(X)\cup \PP(X)|=O(1)$ for every $X$ and $|P_Q|=O(1/\epsilon_0)$
for each $Q\in \BB(X)\cup \PP(X)$, the total space for keeping the
distances $\hat{d}(u,v)$ between vertices and portals is $O(n\log n/\epsilon_0)$. 
By Lemma~\ref{lem-class}, the space for the classes $A^Q_1,..,A^Q_s$ in each 
internal node $X$ is $f(\epsilon_0)$ for every $Q\in \BB(X)\cup \PP(X)$. 
Since there are $O(n)$ internal nodes, the total space for the classes in all 
nodes is $O(nf(\epsilon_0))=O(nf(\epsilon_0))$. 

Therefore the space 
requirement for the oracle is $O(n(\log n/\epsilon_0+f(\epsilon_0)))$.
$\Box$
\end{proof}

\begin{lemma}
For graph $G$ and $\epsilon_0>0$, $\tilde{d}(u,v)$ with
$d_G(u,v)\leq \tilde{d}(u,v)\leq d_G(u,v)+7\epsilon_0d(G)$ can be computed
in $O(1)$ time for any $u$ and $v$ in $G$ using data structure $\DS_0$.
\label{lem-query}
\end{lemma}
\begin{proof}
Let $T_G$ be the recursive subdivision of $G$ in $\DS_0$. $X_u$, 
$X_v$ and $X_{u,v}$ can be found
in $O(1)$ time. If $X_{u,v}$ is a leaf node then $\tilde{d}(u,v)$ can be found in 
$O(1)$ time. Otherwise, for each path $Q\in \BB(X_{u,v})\cup \PP(X_{u,v})$, 
assume that $u\in A^Q_i$ and $v\in A^Q_j$, and let $p^Q_{ij}$ be the key portal for
$A^Q_i$ and $A^Q_j$. By Lemma~\ref{lem-keyportal}, 
\[
\tilde{d}(u,v) = \min_{p^Q_{ij},Q\in \BB(X_{u,v})\cup \PP(X_{u,v})} 
\hat{d}(u,p^Q_{ij})+\hat{d}(p^Q_{ij},v) \leq d_G(u,v)+7\epsilon_0 d(G).
\]
Since $|\BB(X)\cup \PP(X)|=O(1)$ and the key portal $p^Q_{ij}$ can be
found in $O(1)$ time for each path $Q\in \BB(X)\cup \PP(X)$, 
$\tilde{d}(u,v)$ can be computed in $O(1)$ time. 
$\Box$
\end{proof}

From Lemmas~\ref{lem-processing}, \ref{lem-size} and \ref{lem-query}, we have
the following result.
\begin{theorem}
For graph $G$ and $\epsilon_0>0$, there is an oracle which gives a distance 
$\tilde{d}(u,v)$ with 
$d_G(u,v)\leq \tilde{d}(u,v)\leq d_G(u,v)+7\epsilon_0 d(G)$ 
for any vertices $u$ and $v$ in $G$ with $O(1)$ query time, 
$O(n(\log n/\epsilon_0+f(\epsilon_0)))$ size and 
$O(n(\log^3 n/\epsilon_0^2+f(\epsilon_0)))$ pre-processing time.
\label{theo-3} 
\end{theorem}

We can make the oracle in Theorem~\ref{theo-3} a more generalized one: For 
integer $\eta$ satisfying $1\leq \eta\leq 1/\epsilon_0$, we partition each path 
$Q\in \BB(X)\cup \PP(X)$ into $\eta$ segments $Q_1,..,Q_{\eta}$, 
compute the classes $A^{Q_l}_1,..,A^{Q_l}_s$ of vertices in $V(X)$ for each 
segment $Q_l$, $1\leq l\leq \eta$, and key portal $p^{Q_l}_{ij}$, and use
\[
\tilde{d}(u,v)=\min_{p^{Q^l}_{ij},1\leq l\leq \eta, 
Q\in \BB(X)\cup \PP(X)} \hat{d}(u,p^{Q_l}_{ij})+\hat{d}(p^{Q_l}_{ij},v)
\] 
to approximate $d_G(u,v)$. By this generalization, we get the following result.
\begin{theorem}
For graph $G$, $\epsilon_0>0$ and $1\leq \eta\leq 1/\epsilon_0$, there is an 
oracle which gives a distance $\tilde{d}(u,v)$ with 
$d_G(u,v)\leq \tilde{d}(u,v)\leq d_G(u,v)+7\epsilon_0 d(G)$
for any vertices $u$ and $v$ in $G$ with $O(\eta)$ query time,
$O(n(\log n/\epsilon_0+f(\eta\epsilon_0)))$ size and
$O(n(\log^3 n/\epsilon_0^2+f(\eta\epsilon_0)))$ pre-processing time.
\label{theo-4}
\end{theorem}

\section{Oracle with $(1+\epsilon)$ stretch}

For $\epsilon>0$, by choosing an $\epsilon_0=\frac{\epsilon}{7c}$ where $c>0$ is 
a constant, the oracle in Theorem~\ref{theo-3} gives a $(1+\epsilon)$-approximate 
distance oracle for graph $G$ with $d_G(u,v)\geq d(G)/c$ for every $u$ and $v$ in 
$G$. For graph $G$ with $d_G(u,v)$ much smaller than $d(G)$ for some $u$ and $v$,
we use a scaling approach as described in \cite{KST13} to get a 
$(1+\epsilon)$-approximate distance oracle. The idea is to compute a set of 
oracles as described in Theorem~\ref{theo-3}, each for a computed subgraph $H$ of 
$G$. Given $u$ and $v$, we can find in $O(1)$ time a constant number of subgraphs 
(and the corresponding oracles) such that the minimum value returned by these 
oracles is a $(1+\epsilon)$-approximation of $d_G(u,v)$. Therefore a 
$(1+\epsilon)$-approximate distance for any $u,v$ can be computed in constant 
time. We assume $\epsilon>5/n$, otherwise a naive exact distance oracle with 
$O(1)$ query time and $O(n^2)$ space can be used to prove Theorem \ref{theo-1}. 

Let $l_m$ be the smallest edge length in $G$. We assume $l_m\geq 1$ and 
the case where $l_m<1$ can be easily solved in a similar way by normalizing the 
length of each edge $e$ of $G$ to $l(e)/l_m$. For each scale 
$\gamma \in \{2^i|0\leq i\leq \ceil{\log d(G)}\}$, we contract every edge 
$e=\{u,v\}$ of length $l(e)<\gamma/n^2$ in $G$ and then compute a sparse cover 
$\CC_{\gamma}=\{G(\gamma,j),j=1,...,n_{\gamma}\}$ of $G$ as in 
Lemma~\ref{lem-cover}. Each edge of $G$ appears in subgraphs $G(\gamma,j)$ for 
$O(\log n)$ different scales \cite{KST13}. This is because every edge $e$ only 
appears in scales $\gamma$ satisfying $l(e)/24\leq \gamma\leq l(e)n^2$.
The data structure $\DS_1$ for our 
$(1+\epsilon)$-approximate distance oracle keeps the following information:
\begin{itemize}
\item A $2$-approximate distance oracle $\DS_T$ of $G$ in Lemma~\ref{lem-estimate}.
\item Subgraphs $G(\gamma,j)$ and for each subgraph $G(\gamma,j)$, an oracle 
$\DS_0(\gamma,j)$ in Theorem~\ref{theo-3} with $\epsilon_0=\epsilon/c^\prime$, 
$c^\prime>0$ is a constant to be specified below.
\item For every $v\in V(G)$ and every scale $\gamma$, the index $j$ 
of every subgraph $G(\gamma,j)$ that contains $v$.
\end{itemize}
\begin{lemma}
For graph $G$ and $\epsilon>0$, $\tilde{d}(u,v)$ with
$d_G(u,v)\leq \tilde{d}(u,v)\leq (1+\epsilon)d_G(u,v)$ can be computed
in $O(1)$ time for any $u$ and $v$ in $G$ using data structure $\DS_1$.
\label{lem-query1}
\end{lemma}
\begin{proof} 
Given vertices $u$ and $v$ in $G$, oracle $\DS_T$ gives $\tilde{d}_T(u,v)$ with
$d_{G}(u,v)\leq \tilde{d}_T(u,v)\leq 2d_G(u,v)$ in $O(1)$ time 
(Lemma~\ref{lem-estimate}). If $\tilde{d}_T(u,v)=0$ then $0$ is returned as 
$d_G(u,v)$. Otherwise, given $\tilde{d}_T(u,v)$, a scale $\gamma$ with
$\gamma/2 < \tilde{d}_T(u,v)\leq \gamma$ can be found by computing
the most significant %set 
bit of $\ceil{\tilde{d}_T(u,v)}$.
In the word RAM model with unit costs for basic operations, this can be computed 
in $O(1)$ time using the fusion tree data structure proposed in \cite{FW93}. 
%\footnote{We assume that $d(G)/l_m$ can be stored in one computer word of $O(1)$ bits.}.
By Lemma~\ref{lem-cover}, there is a $G(\gamma,j)$ that contains $u$ and 
every $w$ with $d_G(u,w)\leq d_G(u,v)\leq\gamma$ and 
$d(G(\gamma,j))=O(\gamma)=O(d_{G}(u,v))$. Therefore there exists a constant 
$c_1>0$ ($96$ would do) such that $d(G(\gamma,j))\leq c_1 d_G(u,v)$. It is easy 
to see that $\DS_0(\gamma,j)$ returns a minimum distance among all the oracles at 
this scale containing $u,v$. By oracle $\DS_0(\gamma,j)$, we get a distance 
$\tilde{d}_0(u,v)$ with $d_{G(\gamma,j)}(u,v)\leq \tilde{d}_0(u,v)\leq 
d_{G(\gamma,j)}(u,v)+7\epsilon_0d(G(\gamma,j))$. Since $G(\gamma,j)$ is a subgraph 
obtained from $G$ with every edge $e$ with $l(e)<\gamma/n^2$ contracted, 
$d_{G(\gamma,j)}(u,v)\leq d_{G}(u,v)$. Let $L$ be the largest sum of the lengths 
of the contracted edges in any path in $G$. Then
$d_{G}(u,v)\leq d_{G(\gamma,j)}(u,v)+L$ and 
$L<\gamma/n\leq \frac{4}{5}\epsilon d_{G}(u,v)$, from 
$\gamma<2\tilde{d}_T(u,v)\le4 d_{G}(u,v)$ and 
$\epsilon >5/n$. Let $\tilde{d}(u,v)=\tilde{d}_0(u,v)+\gamma/n$. Then
\begin{eqnarray*}
d_{G}(u,v) &\leq & \tilde{d}(u,v)\leq 
d_{G(\gamma,j)}(u,v)+7\epsilon_0d(G(\gamma,j))+\gamma/n\\
&\leq & d_{G}(u,v)+7c_1\frac{\epsilon}{c^\prime} d_{G}(u,v)+\frac{4}{5}\epsilon 
d_{G}(u,v).
\end{eqnarray*}
By choosing $c^\prime=35c_1$, we have 
$d_{G}(u,v)\leq \tilde{d}_{G}(u,v)\leq (1+\epsilon)d_{G}(u,v)$.
By Lemma~\ref{lem-estimate}, it takes $O(1)$ time to compute $\tilde{d}_T(u,v)$.
%The right scale $\gamma$ can be found by computing the most significant set bit
%of $\frac{d(G)}{l_m}$.
From Lemma~\ref{lem-cover},
there are $O(1)$ graphs $G(\gamma,j)$ containing $u$ and $v$. From this and
Theorem~\ref{theo-3}, it takes $O(1)$ time to compute $\tilde{d}(u,v)$. 
$\Box$
\end{proof}

\begin{lemma}
Data structure $\DS_1$ requires $O(n\log n(\log n/\epsilon+f(\epsilon)))$ space
and can be computed in $O(n\log n(\log^3 n/\epsilon^2+f(\epsilon)))$ time.
\label{lem-size1}
\end{lemma}
\begin{proof}
$\DS_T$ requires $O(n\log n)$ space. Each $\DS_0(\gamma,j)$ requires
$O(n_{\gamma j}\log n_{\gamma j}/\epsilon+n_{\gamma j}f(\epsilon))$ space, where
$n_{\gamma j}=|V(G(\gamma,j))|$. Each edge $e$ of
$G$ appears in
$O(\log n)$ different scales $\gamma$ 
and in each scale $e$ appears in $O(1)$ 
subgraphs $G(\gamma, j)$. From this,
$\sum_{\gamma,j} n_{\gamma j}=O(n\log n)$ and $\DS_1$ requires space
$O(n\log n(\log n/\epsilon+f(\epsilon)))$.

$DS_T$ can be computed in $O(n\log^3n)$ time and the sparse neighborhood covers 
can be computed in $O(n\log^2 n)$ time. 
%The value $D$ for computing the classes can be computed in $O(n_{\gamma j})$ time.
The time for computing $\DS_0(\gamma,j)$ for each $G(\gamma,j)$ is
$O(n_{\gamma j}\log^3 n_{\gamma j}/\epsilon^2+f(\epsilon))$. Therefore, 
$\DS_1$ can be computed in $O(n\log n(\log^3 n/\epsilon^2+f(\epsilon))$ time.
$\Box$
\end{proof}

From Lemmas~\ref{lem-query1} and \ref{lem-size1}, we get Theorem~\ref{theo-1}
which is restated below.
\begin{theorem1}
For $\epsilon>0$, there is a $(1+\epsilon)$-approximate distance oracle for $G$ 
with $O(1)$ query time, 
$O(n\log n(\log n/\epsilon+f(\epsilon)))$ size and 
$O(n\log n(\log^3 n/\epsilon^2+f(\epsilon)))$ pre-processing time.
\label{theo-5}
\end{theorem1}

Using the oracle in Theorem~\ref{theo-4} instead of $\DS_0$, we get 
Theorem~\ref{theo-2}.
 
\section{Concluding remarks}

It is open whether there is a $(1+\epsilon)$-approximate distance oracle with
$O(1)$ query time and size nearly linear in $n$ for weighted directed planar
graphs. For undirected planar graphs, it is interesting to reduce oracle size 
and pre-processing time (the function $f(\epsilon)$) for the oracles in this
paper. Experimental studies for fast query time distance oracles are worth
investigating.

%\begin{acknowledgements}
%\end{acknowledgements}



\end{document}